\documentclass[aps,amssymb,amsmath,amsbsy,amsfonts,twocolumn,pra,showpacs]{revtex4}

\usepackage{graphicx}
\usepackage{dcolumn}
\usepackage{bm}
\usepackage{bbm}
\usepackage{psfrag}
\usepackage{epsfig}
\usepackage{amsmath}
\usepackage{amssymb}
\usepackage{color}
\usepackage{mathbbol}
\usepackage{subfigure}
\usepackage{amsthm}
\usepackage{bbold}
\usepackage{latexsym}

\usepackage{hyperref}
\usepackage{graphicx}
\usepackage{graphics}
\usepackage{color}

\usepackage{multirow}
\usepackage{caption}

\newtheorem{theorem}{Theorem}

\newcommand{\beq}{\begin{equation}}
\newcommand{\eeq}{\end{equation}}

\renewcommand{\rho}{\varrho}

\begin{document}

\title{Quantum coherence of Gaussian states}

\author{D. Buono$^{1}$, G. Nocerino$^{1}$, G. Petrillo$^{2}$, G. Torre$^{1}$, G. Zonzo$^{2}$ and F. Illuminati$^{1,3,4}$\footnote{Corresponding author: filluminati@unisa.it}\vspace*{0.2cm}}
\affiliation{$^1$ Dipartimento di Ingegneria Industriale, Universit\`a degli Studi di
Salerno, Via Giovanni Paolo II 132, I-84084 Fisciano (SA), Italy}
\affiliation{$^2$ Dipartimento di Fisica, Universit\`a degli Studi di Salerno,
Via Giovanni Paolo II, I-84084 Fisciano (SA), Italy}
\affiliation{$^3$ Consiglio Nazionale delle Ricerche, Istituto di Nanotecnologia, Rome Unit, I-00195 Roma, Italy}
\affiliation{$^4$ INFN, Sezione di Napoli, Gruppo collegato di Salerno, I-84084 Fisciano
(SA), Italy}

\date{\today}

\begin{abstract}
We introduce a geometric quantification of quantum coherence in single-mode Gaussian states and we investigate the behavior of distance measures as functions of different physical parameters. In the case of squeezed thermal states, we observe that re-quantization yields an effect of noise-enhanced quantum coherence for increasing thermal photon number.
\end{abstract}


\maketitle

\section{Introduction}

The superposition principle is the essential property that discriminates the intrinsically linear and coherent quantum mechanics from the essentially nonlinear and chaotic classical mechanics. Quantum coherence is the key ingredient in all quantum phenomena, from quantum optics and quantum information~\cite{Glauber19630, Scully1991, Albrecht1994,Nielsen2000} to condensed matter physics and quantum thermodynamics~\cite{Deveaud2009,Ford2014, Correa2014, Robnagel2014, Lostaglio2015, Aberg2014}.

In recent years, the scientific community has started a significant effort towards the rigorous exploration of quantum coherence, including its qualification, quantification, and operational significance, along conceptual and mathematical lines analogous to the ones previously followed for the rigorous qualification and quantification of quantum entanglement and quantum correlations.

So far, the investigation have been limited to states of finite-dimensional quantum systems. In this respect, following the seminal identification of entropic and geometric quantifiers and the axiomatic properties that must be satisfied by any proper measure of quantum coherence~\cite{Baumgratz2014}, important progress has included the classification of single-qubit incoherent quantum operations~\cite{2015Streltsov}, the characterization of quantum coherence in two-qubit Bell-diagonal states~\cite{Bromley2015}, and the first results on the operational significance of quantum coherence for quantum technology protocols~\cite{2016Hu}. In particular, it has recently been shown that quantum coherence can be exploited to activate quantum entanglement and other useful quantum correlations~\cite{2016Killoran}.

On the other hand, continuous variable (CV) states of infinite-dimensional systems are of fundamental importance both from theoretical and experimental point of view; in particular a lot of work has been dedicated to Gaussian states and Gaussian channels, as they are both theoretically easy to manage and experimentally easy to produce. Moreover, although non-Gaussian CV states have been proved quite useful in particular cases, such as CV quantum teleportation~\cite{2007Dellanno, 2010Dellanno, 2010Dellanno2}, Gaussian states constitute the foundation of every quantum processing tasks in the CV setting. As a consequence, it is important to generalize the quantification of quantum coherence and the classification of quantum incoherent operations to the case of Gaussian states.

In this work we follow the geometric approach pioneered by Baumgratz, Cramer, and Plenio~\cite{Baumgratz2014}, and we compute different distance measures in detail for the main classes of single-mode Gaussian states. 

The article is organized as follows. In Sec.~\ref{sec1} we briefly review the basic formalism of single-mode Gaussian states. In Sec.~\ref{sec2} we classify the set of incoherent single-mode Gaussian states, namely those states that have diagonal density matrix in the Fock basis, showing
that they correspond to purely thermal states with zero displacement and zero squeezing.
In Sec.~\ref{sec3} we define quantum coherence as the minimal distance from the set of thermal states and we introduce two such geometric measures, respectively in terms of the Bures and Hellinger metric. In Sec.~\ref{sec4} we explicitly compute the given measures for the most significant classes of single-mode Gaussian states as functions of different physical parameters such as displacement, squeezing, and thermal noise. In particular, we show that squeezed thermal states feature a remarkable phenomenon of noise-enhanced quantum coherence with increasing number of thermal photons at fixed squeezing. In Sec.~\ref{sec5}, we draw our conclusions and discuss future research on the extension to multimode Gaussian states and the application to CV quantum technologies. In the App.~\ref{App1} we include various mathematical and calculational details leading to the results presented in the main text.

\section{Single-mode Gaussian states}
\label{sec1}

In this section we briefly review the mathematical formalism aiming at describing one-mode Gaussian states.\\
Every single-mode Gaussian state can be expressed in the following form~\cite{2005Ferraro, 2011Weedbrook, 2014Adesso}
\begin{equation}
\rho_G=D(\beta)S(\xi)\nu_{th}S^\dagger(\xi)D^\dagger(\beta),
\label{gensinglemode}
\end{equation}
where $D(\beta)=\exp \mathcal{f} \beta a^\mathcal{y}-\beta^*a \mathcal{g}$ is the displacement operator with complex amplitude $\beta$, $S(\xi)=\exp \mathcal{f} \frac{1}{2}\xi (a^\mathcal{y})^2-\frac{1}{2}\xi^*a^2\mathcal{g}$ is the squeezing operator with squeezing parameter $\xi=re^{i\psi}$ and $\nu_{th}$ in the one-mode thermal state:
\begin{equation}
\nu_{th}=\frac{1}{1+N_{th}}\sum_{n=0}^\mathcal{1} \left( \frac{N_{th}}{1+N_{th}} \right)^n \left| n \right \rangle \left \langle n \right |.
\end{equation}

Due to the Gaussian form of the Wigner function representing these states in the phase space, they are characterized uniquely by the first moments (the displacement vector):

\begin{equation}
R\equiv ( \left \langle X \right \rangle,\left \langle P \right \rangle)=\sqrt{2}(Re[\beta],Im[\beta])
\label{firstm}
\end{equation}
and second moments (the covariance matrix):
\begin{equation}
\sigma_G=\begin{bmatrix}
a & c  \\
c & b
\end{bmatrix}\;.
\label{normalform}
\end{equation}
with
\begin{align}
\label{abc}
a &=\frac{1+2N_{th}}{2}(\cosh (2r)+\cos \psi \sinh (2r)),\nonumber \\
b &=\frac{1+2N_{th}}{2}(\cosh (2r)-\cos \psi \sinh (2r)), \\
c &= \frac{1+2N_{th}}{2}(\sin \psi\sinh (2r)) \nonumber .
\end{align}

This parametrization will be useful for the study of quantum coherence for the classes of Gaussian state considered in Sec.~\ref{sec4}.

\section{Incoherent single-mode Gaussian states}
\label{sec2}

In this section we identify the set of one-mode Gaussian incoherent states. Indeed, we obtain the following result:

\begin{theorem}
A single mode Gaussian state is incoherent if and only if the covariance matrix $\sigma$ is diagonal, there is not displacement and squeezing.
\end{theorem}

\begin{proof}
By definition ~\cite{Baumgratz2014}, a single mode incoherent state is represented by a diagonal density matrix $\rho_I$,

\begin{equation}
\rho_I=\sum_{n}p_n \left | n \right \rangle \left \langle n \right |
\label{incdensmatr}
\end{equation}

It is easy to verify that the corresponding {\em{CM}} $\sigma_I$ is diagonal with elements
\begin{equation}
\sigma_{jj}=\sum_{n}p_n \frac{2n+1}{2}
\end{equation}
with $j=1,2$, while the first moments $\left \langle X \right \rangle$ and $\left \langle P \right \rangle$ are null, so that the displacement is zero. So if the state is incoherent, that is described by the density matrix $\rho_I$, the corresponding {\em{CM}} $\sigma_I$ is such that the elements have to satisfy the following conditions:
\begin{eqnarray}
\left \langle X \right \rangle &=& 0 =\left \langle P \right \rangle,\label{cond1} \\
\sigma_{11}&=&\sigma_{22}=\sum_{n}p_n \frac{2n+1}{2}, \label{cond2}\\
\sigma_{12}&=&\sigma_{12}=0\label{cond3}.
\end{eqnarray}
Referring explicitly to the generic single mode Gaussian state Eq. (\ref{gensinglemode}), the condition of incoherence Eq. (\ref {cond1}) is satisfied by null values of the complex amplitude $\beta$. Indeed the vector of the first moments $R$ Eq. (\ref {firstm}) is null just for $\beta=0$ (no displacement). Similarly, the conditions Eq.s (\ref {cond2}) and (\ref {cond3}) are satisfied by the values $r=0$ and $\psi=0$ (no squeezing), for which $a=b=\frac{1+2N_{th}}{2}\cosh(2r)$ and $c=0$. This means that a single mode Gaussian state is incoherent when there is no squeezing and no displacement thermal state. Under these conditions the generic Gaussian state Eq. (\ref {gensinglemode}) becomes thermal one.
So, we have shown that a single mode Gaussian state is incoherent if it is a thermal state, $\rho_I\equiv \nu_{th}$ with $p_n=\frac{N_{th}^n}{(1+N_{th})^ {(n+1)}}$.
Let consider now a Gaussian state described by a diagonal {\em{CM}} and without squeezing and displacement. It easy to verify that only in this case the density matrix describing the state is diagonal, $ \left \langle m \left | \rho_G \right |n \right \rangle=0$. Indeed a density matrix represented by a diagonal {\em{CM}}  and by at least no null value between $\beta$ and $r$ isn't diagonal, $ \left \langle m \left | \rho \right |n \right \rangle \neq0$ just as  a the density matrix of a state with $\beta=r=0$ and not diagonal {\em{CM}}.
\end{proof}
In conclusion, the only incoherent single mode Gaussian states are the thermal states, described by the CM:
\begin{equation}
\sigma_{i}=\frac{1}{2}\begin{bmatrix}
1+2N_{th,i} & 0  \\
0 & 1+2N_{th,i}
\end{bmatrix}\;.
\label{IncCM}
\end{equation}

This class of states will be the reference set from which it is necessary to calculate the minimum distance.

\section{Quantum coherence of Gaussian states: distance from the set of incoherent states}
\label{sec3}

We extend the coherence measure $C_x$ introduced in \cite{Baumgratz2014} to Gaussian states $\rho_G$ Eq. (\ref{gensinglemode}), considering the minimum of the distance, induced by the $x$-metric, from the set of the incoherent states Eq. (\ref{incdensmatr}):
\begin{equation}
C_x=\frac{1}{2} \min_{\rho_I}d^2_x(\rho_G,\rho_I).
\end{equation}
In particular, we consider the {\em{Bures}} and the {\em{Hellinger}} metric. \\
The Bures measure is then defined in terms of the Bures distance as
\begin{equation}
C_{Bu} \equiv \frac{1}{2} \min_{\rho_I}d^2_{Bu}(\rho_G,\rho_I)\\
= 1-\max_{\rho_I} \sqrt{\mathcal {F}(\rho_G,\rho_I)},
\label{BuresCoe}
\end{equation}
where $\mathcal{F}(\rho_G,\rho_I)$ is the Uhlmann fidelity.
For Gaussian states, the Uhlmann fidelity~\cite{2012Marian} is found to be:
\begin{equation}
\label{pagfid}
\mathcal {F}=e^{-\frac{1}{2}\delta R^T(\sigma_G+\sigma_I)^{-1}\delta R}(\sqrt{\Delta+\Lambda}-\sqrt{\Lambda})^{-1},
\end{equation}
where $\delta R$ is the difference between the first moments vectors of the relative states $\rho_G$ and $\rho_I$. In the previous section we showed that the class of the incoherent states has null first moments; hence $\delta R\equiv R$ Eq. (\ref{firstm}). The explicit expressions of $\Lambda$, and $\Delta$ are:
\begin{equation}
\Delta = \det(\sigma_G + \sigma_I),
\end{equation}
\begin{equation}
\Lambda = 4 \det(\sigma_G + \dfrac{i}{2} \Omega) \det(\sigma_I + \dfrac{i}{2} \Omega),
\end{equation}
where $ \Omega $ is the symplectic matrix:
\begin{equation}
\Omega=\begin{bmatrix}
0 & 1  \\
-1 & 0
\end{bmatrix}\;.
\end{equation}
The measure Eq. (\ref{BuresCoe}) enjoys the main properties that a {\em{bona fide}} coherence measure must satisfy according to \cite{Baumgratz2014}. Indeed, from the properties of the Fidelity~\cite{1976Uhlmann}, the necessary requirements $(C1')$, $(C2a)$ and $(C3)$ of \cite {Baumgratz2014} are automatically satisfied by the definition Eq.~(\ref{BuresCoe}). In fact, $C_{Bu}=0$ iff $\rho_G \equiv \rho_I$ $(C1')$; moreover $C_{Bu}$ is monotone under incoherent completely positive trace preserving map {\em{ICPTP}} $(C2a)$; finally $C_{Bu}$ is convex because $d_{Bu}$ is contractive and $d_{Bu}^2$ is convex $(C3)$. \\
In the same way, the Hellinger measure is defined in term of the Hellinger distance as:
\begin{equation}
\label{Hellmea}
C_{He} \equiv \frac{1}{2} \min_{\rho_I}d^2_{He}(\rho_G,\rho_I)\\
= 1-\max_{\rho_I} \sqrt{\mathcal {A}(\rho_G,\rho_I)},
\end{equation}
where $ \mathcal{A}(\rho_G,\rho_I) $ is the Affinity~\cite{Marian2015}. For Gaussian states, the Affinity assumes the simple form:
\begin{equation}
\mathcal{A}(\rho_G,\rho_I)=2 e^{-\frac{1}{2}\delta R^T(\sigma_G+\sigma_I)^{-1}\delta R} \dfrac{[\det(\sigma_G) \; \det(\sigma_I)]^{\frac{1}{4}}}{[\det(\sigma_G + \sigma_I)]^{\frac{1}{2}}}.
\end{equation}
It is possible to show, following the same line of reasoning given for the Bures distance, that the Hellinger distance enjoys the required properties.

\section{Results}
\label{sec4}

In this section we analyse the behaviour of the coherence measures Eqs.~(\ref{BuresCoe}) and~(\ref{Hellmea}) for the principal classes of one-mode Gaussian states, obtained from Eqs.~(\ref{normalform}) and~(\ref{abc}) for particular values of parameters.

\subsection{Squeezed thermal states \label{GaussianStates}}

Let us consider, the set of squeezed thermal states (STS), with real squeezing parameter $\xi\equiv r$. Such states is described by the CM (\ref{normalform}) with $a=\left(\frac{1 + 2N_{th}}{2}\right)\left(\cosh (2r)+ \sinh (2r)\right)$, $b=\left(\frac{1 + 2N_{th}}{2}\right)\left(\cosh (2r)- \sinh (2r)\right)$ and  $c=0$.


\begin{figure}[h]
\centering\includegraphics[width=0.45\textwidth]{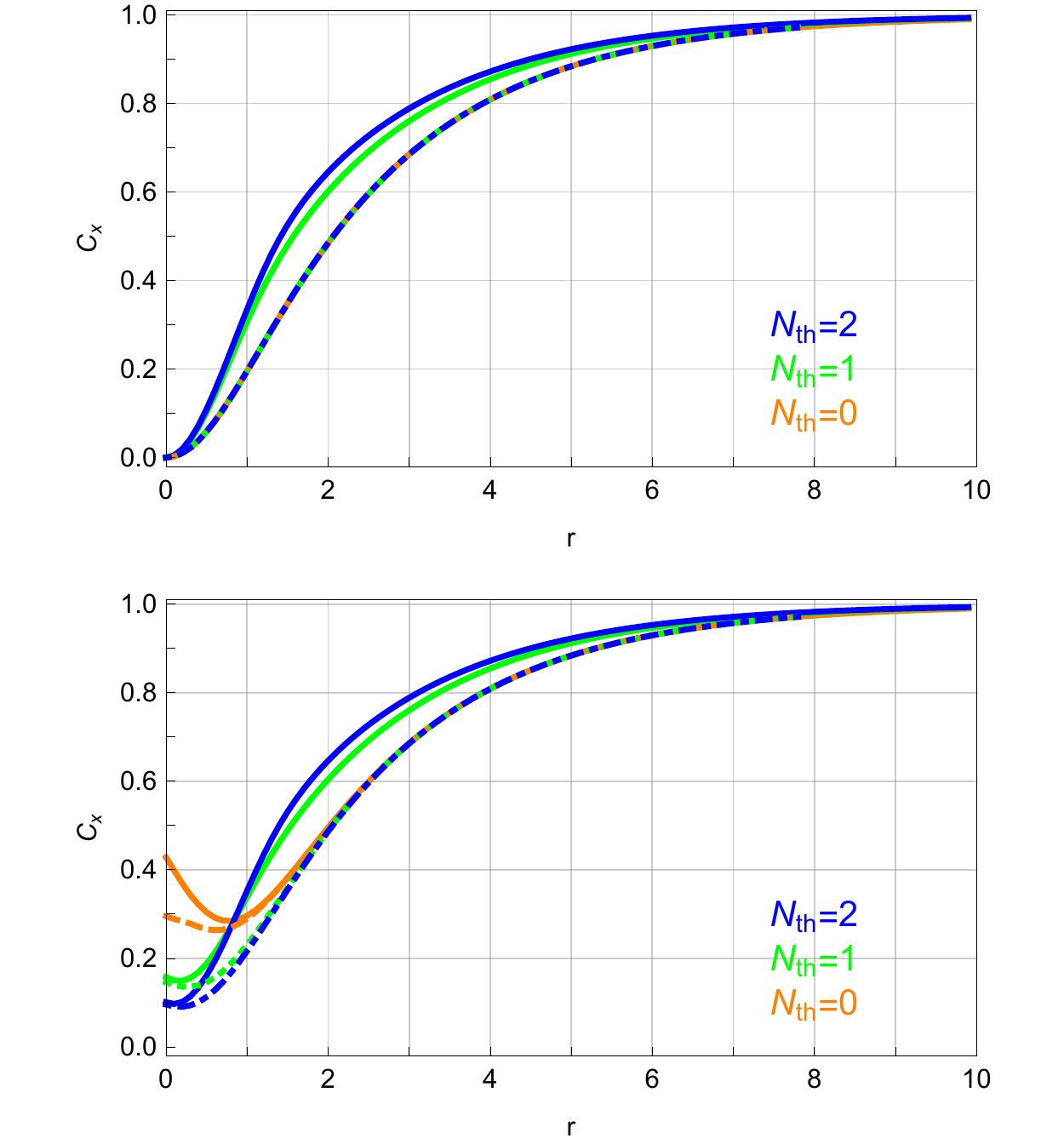}
\caption{(color online) Behavior of the coherence measures Eqs.~(\ref{BuresCoe}) and~(\ref{Hellmea}) for STSs, as a function of the squeezing parameter $r$ for fixed $N_{th}=0$ (orange), $N_{th}=1$ (green), $N_{th}=2$ (blue). The solid and dashed lines represent the Bures and Hellinger measures respectively. The case  $ \beta = 0 $ is shown in the upper panel, while in the lower panel the case $ \beta = 1 $ is reported.}
 \label{GrSTSvsr}
\end{figure}

In Fig.~(\ref{GrSTSvsr}) is reported the behaviour of the coherence measures Eqs.~(\ref{BuresCoe}) and~(\ref{Hellmea}) as a function of the squeezing parameter $r$ for the squeezed vacuum state ($N_{th}=0$) (orange line), and for fixed $N_{th}=1$ (green line) and $N_{th}=2$ (blue line) for the cases $ \beta = 0 $ (upper panel) and $ \beta = 1 $ (lower panel). The coherence measures increases with $r$, due to the fact that the state becomes more and more distinguishable form the set of the matrices proportional to the identity matrix, namely, the set of thermal (incoherent) states.
Therefore:
\begin{equation}
C_d(\tilde{r, N_{th}}) \rightarrow 1 \;\;\; \text{for} \;\;\; r \rightarrow \infty.
\end{equation}

\begin{figure}[h]
\centering\includegraphics[width=0.45\textwidth]{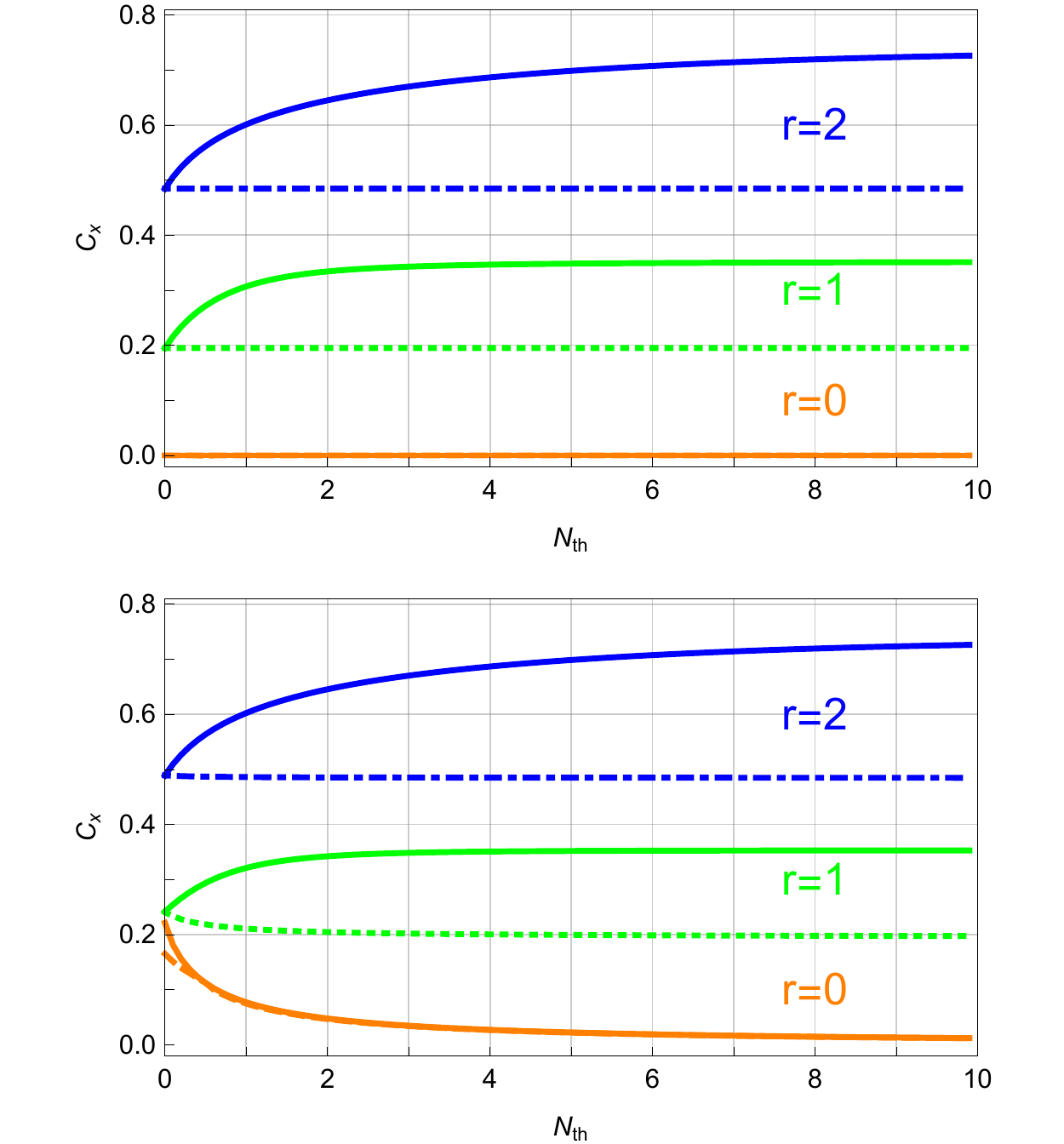}
\caption{(color online) Behavior of the coherence, for STSs, as a function of the thermal photon number $N_{th}$ for fixed $r=0$ (orange),
$r=1$ (green) and $r=2$ (blue). The solid and dashed lines stand for Bures and Hellinger measures respectively. In the upper panel is shown the case $ \beta = 0 $, in the lower panel $ \beta = 1 $.}
\label{GrSTSvsNth}
\end{figure}

In Fig. (\ref{GrSTSvsNth}) is reported the behaviour of quantum coherence measures Eqs.~(\ref{BuresCoe}) and~(\ref{Hellmea}) as a function of the mean thermal photon number $N_{th}$ of the STS for $r=0$  (orange line), for $r=1$ (green line) and for $r=2$ (blue line). The value of the coherence tends asymptotically to a constant value, as shown in~(\ref{App1}).
As you can see from in Fig. (\ref{GrSTSvsNth}), the value of $ C_d $ tends asymptotically to a constant (in the case of the Hellinger, is constant for each $ N_{th} $). Apparently, this may seem abnormal behavior, since the elements $ a $ and $ b $ of the CM always differ more increasing of $ N_{th} $. However, we have shown analytically (see Appendix A) that the $ C_d $ value tends asymptotically to a constant for $ N_{th} \rightarrow \infty $. Indeed, although the terms $ a $ and $ b $ differ more and more, for very large $ N_{th} $ both tend to infinite of the same order. Therefore:
\begin{equation}
C_d(\tilde{r, N_{th}}) \rightarrow \text{const} \;\;\; \text{for} \;\;\; N_{th} \rightarrow \infty
\end{equation}

\begin{center}
\begin{tabular}{|c|c|c|}
\hline
\multicolumn{3}{|c|}{Squeezed Thermal State} \\
\hline
$ C_x $ & $ N_{sq} $ & $ N_{th} $ \\
\hline
$ C_{Bu} $ & $ 10^{11} - 10^{12} $ & 2 \\
\hline
$ C_{He} $ & 2 & $ 10^{11} - 10^{12} $ \\
\hline
\end{tabular}
\captionof{table}{Example values of the parameters $ N_{sq} $ and $ N_{th} $ for the coherence measures Eq.~(\ref{BuresCoe}), and Eq.~(\ref{Hellmea}) corresponding to $ 99 \% $ of quantum coherence.}
\label{tabsts}
\end{center}

In Tab.(\ref{tabsts}) some value of parameters corresponding to $ 99 \% $ of quantum coherence are reported for the quantum coherence measures considered in the text.

\subsection{Coherent thermal states }

We consider, at first, coherent thermal states CTSs with real amplitude of the displacement so that the first moments Eq. (\ref{firstm}) are $R=\sqrt {2}(\beta,0)$, while the CM corresponds exactly to the covariance matrix of a thermal state,
\begin{equation}
\sigma_{CTS}=\frac{1}{2}\begin{bmatrix}
1+2N_{th} & 0  \\
0 & 1+2N_{th}
\end{bmatrix}\;.
\label{CTSCM}
\end{equation}

\begin{figure}[h]
\centering\includegraphics[width=0.45\textwidth]{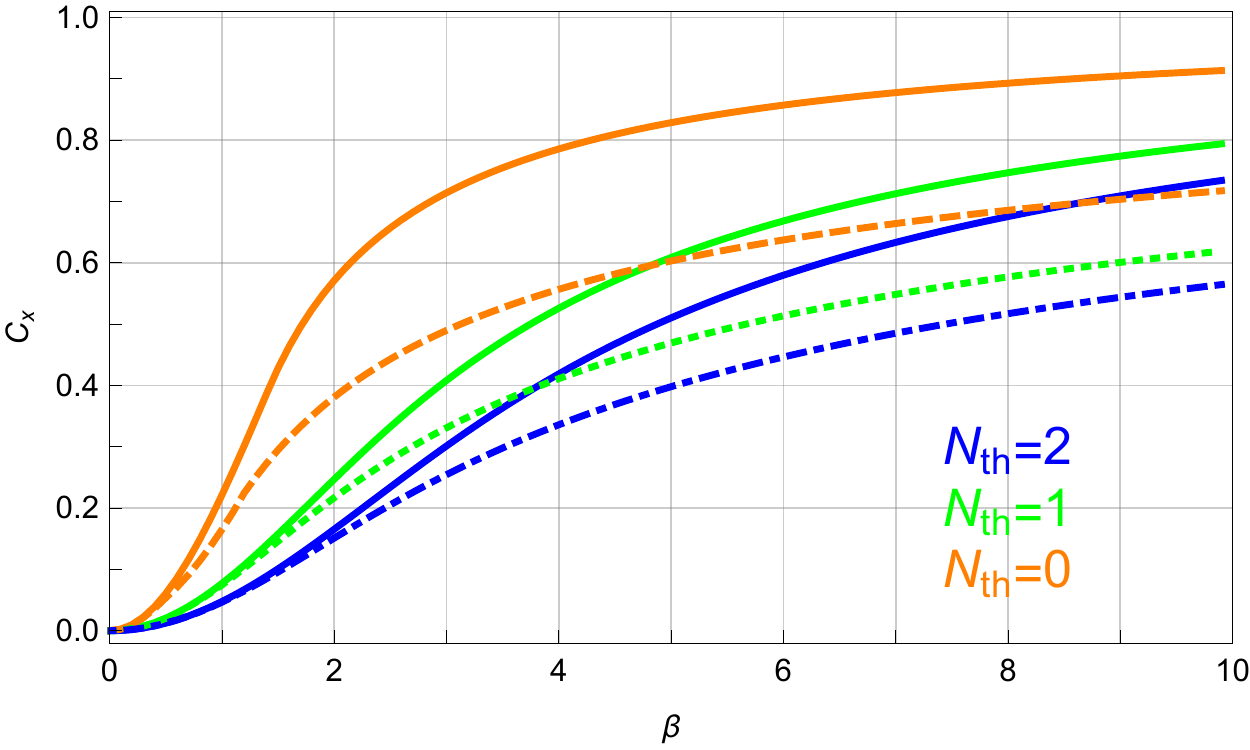}
\caption{(color online) Behavior of coherence, for a CTS, as a function of the displacement $\beta$, for different values of the mean thermal
photon number: $N_{th}=0$ (orange), $N_{th}=1$ (green) and $N_{th}=2$ (blue). The solid and dashed lines stand for Bures and Hellinger measures
respectively.}
\label{GrCTSvsN}
\end{figure}

In Fig (\ref{GrCTSvsN}) it is reported the behaviour of quantum coherence measures Eqs.~(\ref{BuresCoe}) and~(\ref{Hellmea}) as a function of $\beta$ for CTS with fixed thermal photon number $N_{th}=0$ (orange line), representing the ideal coherent states, $N_{th}=1$ (green line), and $N_{th}=2$ (blue line). As expected, the coherence increases at the increasing of the displacement. In particular, fixing the thermal contribution to $N_{th}\equiv \tilde{N}_{th}$, it results:

\begin{equation}
C_{x}(\beta,\tilde{N}_{th}) \rightarrow 1 \;\;\; \text{for} \;\;\; \beta \rightarrow \infty.
\end{equation}

\begin{figure}[h]
\centering\includegraphics[width=0.45\textwidth]{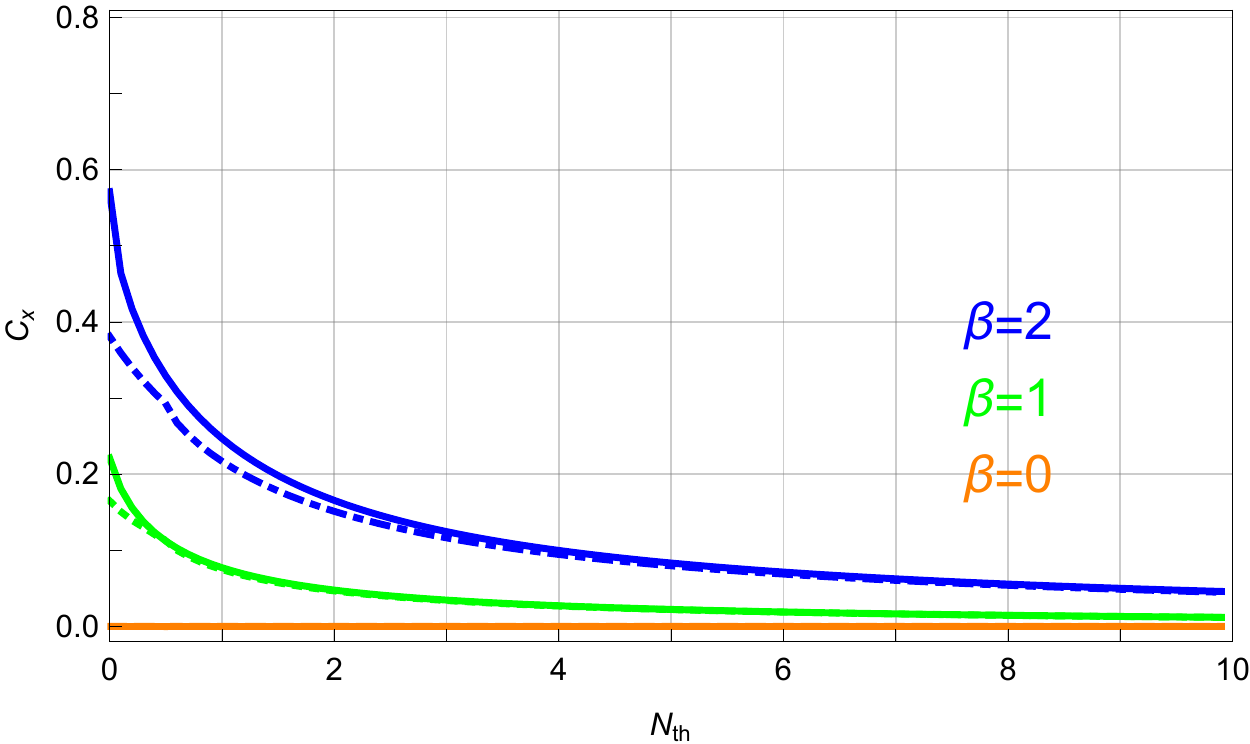}
\caption{(color online) Behavior of coherence, for a CTS, as a function of the thermal photon number $N_{th}$, for different values of displacement: $\beta=0$ (orange), $\beta=1$ (green) and $\beta=2$ (blue). The solid and dashed lines stand for Bures and Hellinger measures
respectively.}
\label{GrCTSvsNth}
\end{figure}

In Fig. (\ref{GrCTSvsNth}) is reported the behaviour of quantum coherence measures Eqs.~(\ref{BuresCoe}) and~(\ref{Hellmea}) as a function of the thermal contribution $N_{th}$ for three different values of the coherent amplitude: $\beta=0$ (orange line), that represents thermal (incoherent) state, $\beta=1$ (green line) and $\beta=2$ (blue line). The measures are decreasing at the increasing of $N_{th}$. We can see that, more in general, for a generic fixed value of the coherent amplitude $\beta \equiv \tilde{\beta}$,


\begin{equation}
C_{d}(\tilde{\beta},{N}_{th}) \rightarrow 0 \;\;\; \text{for} \;\;\; {N}_{th} \rightarrow \infty.
\end{equation}

This is due to the fact that the contribution of the displacement $ \beta $ to the quantum coherence becomes more and more negligible at $ N_{th} $ increasing, namely the number of incoherent thermal photons $ N_{th,i} $, that minimize the measures, tends to the number of thermal photons of the state.

\begin{center}
\begin{tabular}{|c|c|c|}
\hline
\multicolumn{3}{|c|}{Coherent Thermal State} \\
\hline
$ C_x $ & $ N_{coh} $ & $ N_{th} $ \\
\hline
\multirow{2}*{$ C_{Bu} $} & $ 1.6 \times 10^3 $ &  $2$\\
\cline{2-3}
& $2$ & never \\
\hline
\multirow{2}*{$ C_{He} $} & $ 10^4 $ & 2\\
\cline{2-3}
& $2$ & never \\
\hline
\end{tabular}
\captionof{table}{Example values of the parameters $ N_{coh} $ and $ N_{th} $ for the coherence measures Eq.~(\ref{BuresCoe}), and Eq.~(\ref{Hellmea}) corresponding to $ 99 \% $ of quantum coherence.}
\label{tabcts}
\end{center}

In Tab.(\ref{tabcts}) some value of parameters corresponding to $ 99 \% $ of quantum coherence are reported for the quantum coherence measures considered in the text.

\subsection{Thermal squeezed states}
In conclusion, we consider thermal squeezed states TSSs with real squeezing parameter $\xi\equiv r$. Such state is described by the CM
\begin{equation}
\sigma_{TSS}=\begin{bmatrix}
N_{th} + e^{2 r} & 0  \\
0 & N_{th} + e^{-2 r}
\end{bmatrix}\;.
\label{TSSCM}
\end{equation}

\begin{figure}[h]
\centering\includegraphics[width=0.45\textwidth]{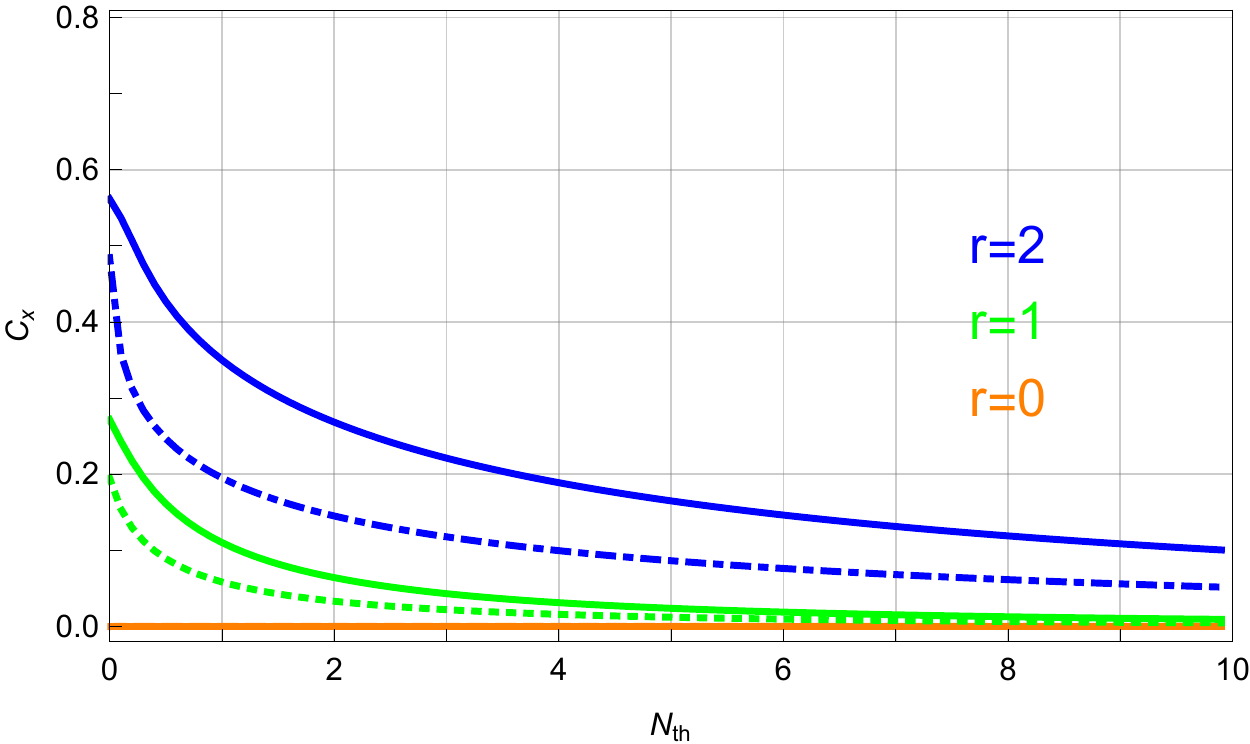}
\caption{(color online) Behavior of coherence, for a TSS, as a function of the thermal photon number $N_{th}$, for different values of squeezing: $r=0$ (orange), $r=1$ (green) and $r=2$ (blue). The solid and dashed lines stand for Bures and Hellinger measures, respectively.}
\label{GrTSSvsNth}
\end{figure}

In Fig. (\ref{GrTSSvsNth}) is reported the behaviour of quantum coherence measures Eqs.~(\ref{BuresCoe}) and~(\ref{Hellmea}) as a function of the thermal contribution $N_{th}$ for three different values of the squeezing parameter: $r=0$ (orange line), $r=1$ (green line) and $r=2$ (blue line). The measures are decreasing at the increasing of $N_{th}$. Indeed, fixing the value $ r $ of the squeezing, increasing of the thermal photons $ N_{th} $ causes the diagonal elements of the covariance matrix Eq.~(\ref{TSSCM}) to become less distinguishable. As a consequence, in the limit $ N_{th} \rightarrow \infty $ the TSS become incoherent state. Hence we have:
\begin{equation}
C_{d}(\tilde{r},{N}_{th}) \rightarrow 0 \;\;\; \text{for} \;\;\; {N}_{th} \rightarrow \infty.
\end{equation}

\begin{figure}[h]
\centering\includegraphics[width=0.45\textwidth]{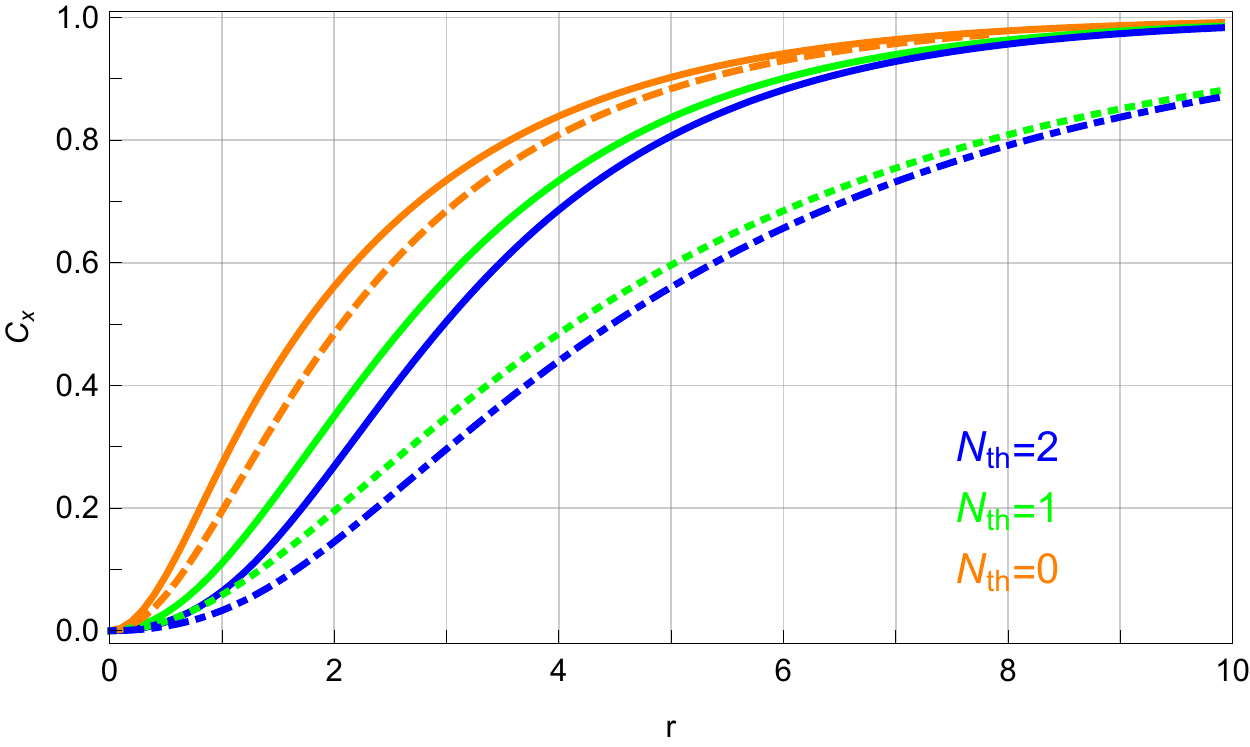}
\caption{(color online) Behavior of coherence, for a TSS, as a function of the squeezing $ r $, for different values of thermal photon number: $N_{th}=0$ (orange), $N_{th}=1$ (green) and $N_{th}=2$ (blue). The solid and dashed lines stand for Bures and Hellinger measures, respectively.}
\label{GrTSSvsr}
\end{figure}

In Fig. (\ref{GrTSSvsr}) is reported the behaviour of quantum coherence measures Eqs.~(\ref{BuresCoe}) and~(\ref{Hellmea}) as a function of the squeezing parameter $ r $ for three different values of thermal photon number: $N_{th}=0$ (orange line), $N_{th}=1$ (green line) and $N_{th}=2$ (blue line). The measures are increasing at the increasing of $ r $. \\
At $ N_{th} $ fixed, increasing $ r $ causes the diagonal elements of $ \sigma_{TSS} $ to become more distinguishable. As a result the covariance matrix defining the state becomes more and more different from the matrices that define the incoherent state.
Consequently, it is never possible to find a value of $ N_{th,i} $ that minimizes the distance between $ \sigma_{TSS} $ and $ \sigma_{I} $. Therefore we have:
\begin{equation}
C_{d}(r,\tilde{N}_{th}) \rightarrow 1 \ \;\;\; \text{for} \;\;\; r \rightarrow \infty.
\end{equation}

\begin{center}
\begin{tabular}{|c|c|c|}
\hline
\multicolumn{3}{|c|}{Thermal Squeezed State} \\
\hline
$ C_x $ & $ N_{sq} $ & $ N_{th} $ \\
\hline
\multirow{2}*{$ C_{Bu} $} & $ 2 $ &  never\\
\cline{2-3}
& $10^{12} - 10^{13}$ & $2$ \\
\hline
\multirow{2}*{$ C_{He} $} & $2$ & never\\
\cline{2-3}
& $10^{16} - 10^{17}$ & $2$ \\
\hline
\end{tabular}
\captionof{table}{Example values of the parameters $ N_{sq} $ and $ N_{th} $ for the coherence measures Eq.~(\ref{BuresCoe}), and Eq.~(\ref{Hellmea}) corresponding to $ 99 \% $ of quantum coherence.}
\label{tabtss}
\end{center}

In Tab.(\ref{tabcts}) some value of parameters corresponding to $ 99 \% $ of quantum coherence are reported for the quantum coherence measures considered in the text.

\section{Conclusions}
\label{sec5}

Summarizing we have defined a geometric-based approach for the quantification of quantum coherence for one-mode Gaussian states as the minimal distance of the given state form the set of incoherent (thermal) states. Furthermore, we have studied the behaviour of the quantum coherence for the main classes of the quantum Gaussian states. \\

As expected, the increase of the displacement $ \beta $ causes the increase of the quantum coherence measures. Instead, in general, increasing the number of thermal photons, has a detrimental effect on the quantum coherence. Surprisingly, we observe that, for the squeezed thermal states, except for the case $ r=0 $, the quantum coherence tends to a constant value different from zero for $ N_{th} \gg 1 $; furthermore for certain values of the squeezing parameter $ r $ this asymptotic value is greater than the initial value ($ N_{th} = 0 $). This counterintuitive behaviour is the same as observed for other geometric quantum correlations measure~\cite{2015Roga}.

This work represents the first step in understanding the usefulness of quantum coherence for the quantum technologies. The natural next step will be to extend the measure to the case of two-mode Gaussian states~\cite{2016Torre} and to classify the Gaussian incoherent quantum operations.

\appendix

\section{Asymptotic study of the Bures measure for the $ \sigma_{STS} $}
\label{App1}

In this section we show that, for the class of squeezed thermal states, the value of coherence, when the number of thermal photons is sufficiently high, tends to a squeezed-dependent value.

We consider, without loss of generality, a covariance matrix of the form
\begin{equation}
\label{testcm}
\sigma_p=\begin{bmatrix}
C_1 \left( \frac{1 + 2 N_{th}}{2} \right) & 0  \\
0 & C_2 \left( \frac{1 + 2 N_{th}}{2} \right)
\end{bmatrix}\;,
\end{equation}
as Eq.~(\ref{testcm}) simulates perfectly the state $ \sigma_{STS} $ in the case of fixed squeezing $ r $, putting:
\begin{align}
C_1 = (\sinh(r) + \cosh(r)) \nonumber \\
C_2  =  (\cosh(r) - \sinh(r)).
\end{align}
We consider, furthermore, the covariance matrix Eq.~(\ref{IncCM}) of the reference states.\\

For ease of calculation we fix the displacement $ \beta = 0 $; we have verified numerically that the more general case presents the same behaviour. With this choice, the exponential in the fidelity is equal to one, namely Eq.~(\ref{pagfid}) reduces to:
\begin{equation}
\mathcal {F}=\dfrac{1}{\sqrt{\Delta+\Lambda}-\sqrt{\Lambda}};
\end{equation}

\begin{figure}[h]
\centering\includegraphics[width=0.45\textwidth]{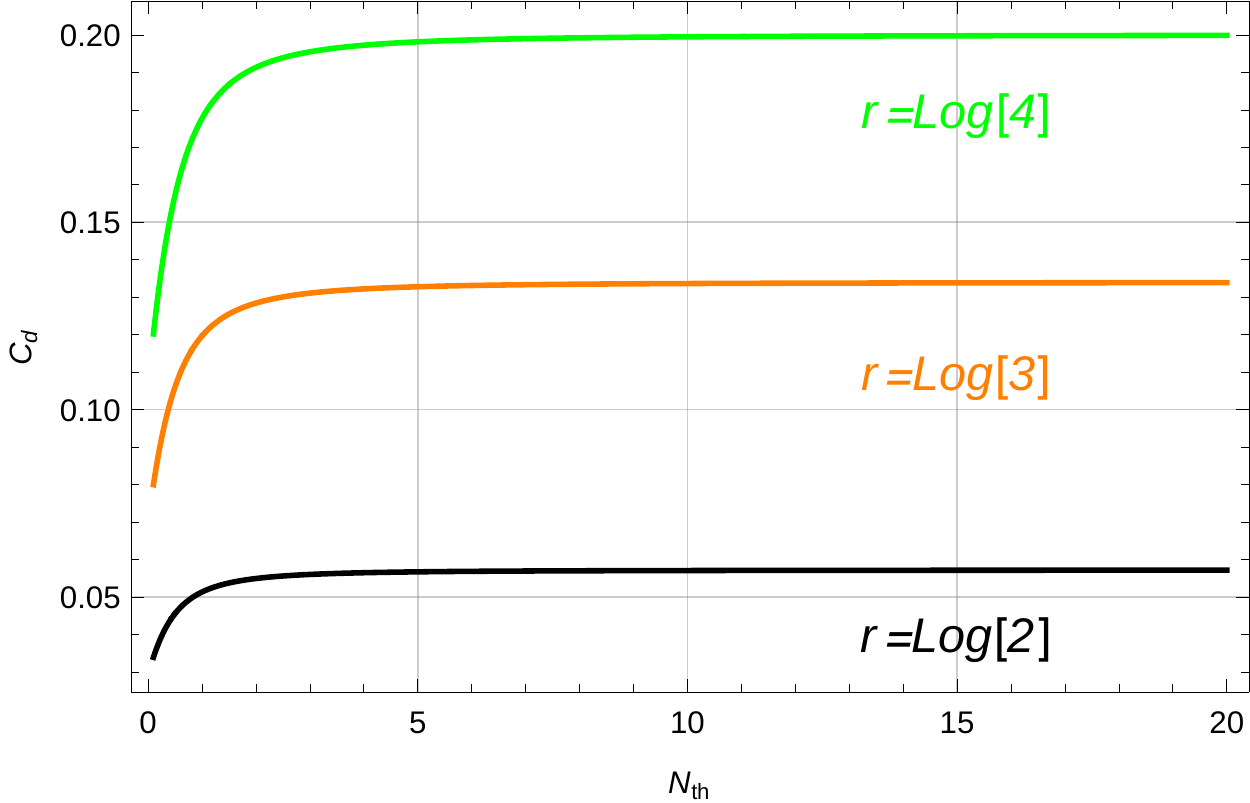}
\caption{(color online) Asymptotic behavior of the Bures measure for squeezed thermal state for different values of the squeezing parameter $r$. }
\label{fidel_r}
\end{figure}

Now we need to maximize the fidelity on the set of incoherent states, namely, we have to maximize on $ N_{th,i} $. Remember that the number of thermal photons $ N_{th} $ and $ N_{th,i} $ are greater or equal than $ 0 $, it is possible to show that the maximized fidelity is given by:
\begin{equation}
\label{fidan}
\mathcal{F}_{a,r}= \left\lbrace \begin{array}{cc}
N(r) & a=0 \\
\mathcal{PZ}_{a,r} [1] & 0<a< \bar{a}_1(r) \\
\mathcal{PZ}_{a,r} [3]  & \bar{a}_1(r)\leq a < \bar{a}_2(r) \\
\mathcal{PZ}_{a,r} [2]  & a \geq \bar{a}_2(r)
\end{array} \right. ,
\end{equation}\\

where $ 0 \leq N(r) \leq 1 $ is a $ r $-dependent constant and $ \mathcal{PZ}_{a,r} [i] $ are the $ i$-th solution of the following polynomial $ \mathcal{PZ}_{a,r} (\zeta) $:
\begin{widetext}
\begin{align}
\label{petr2}
\mathcal{PZ}_{a,r} (\zeta) &= (K_8(r) a^8+K_7(r) a^7+K_6(r) a^6 +K_5(r) a^5+K_4(r) a^4+K_3(r) a^3 +K_2(r) a^2+K_1(r) a+K_0(r) )\zeta^8  \nonumber \\
&+(I_8(r) a^8+I_7(r) a^7+I_6(r) a^6 +I_5(r) a^5+I_4(r) a^4 +I_3(r) a^3 +I_2(r) a^2+I_1(r) a+I_0(r))\zeta^6 \nonumber \\
&+(J_8(r) a^8+J_7(r) a^7+J_6(r) a^6+J_5(r) a^5+J_4(r) a^4+J_3(r) a^3+J_2(r) a^2+J_1(r) a+J_0(r))\zeta^4 \nonumber \\
&+(L_6(r) a^6+L_5(r) a^5+L_4(r) a^4+L_3(r) a^3+L_2(r) a^2+L_1(r) a+L_0(r))\zeta^2 \nonumber \\
&+M_4(r) a^4+M_3(r) a^3+M_2(r) a^2+M_1(r) a+M_0(r)
\end{align}
\end{widetext}
From the explicit expression Eqs.~(\ref{fidan}) and~(\ref{petr2}) it is possible to show that the asymtotic value of the Fidelity with respect the number of thermal photons is a $ r $-dependent positive constant:
\begin{equation}
\lim_{a \rightarrow \infty} \mathcal{F}_{a,r} = \lim_{a \rightarrow \infty} \mathcal{PZ}_{a,r} [2] = \text{const}.
\end{equation}
To better clarify the result, in Fig.~(\ref{fidel_r}) we show the behaviour of the Bures measure as a function of $ N_{th} $ for different values of the squeezing parameter $ r $.


\label{apphell}

\

\end{document}